\documentclass[conference]{IEEEtran}
\usepackage{amsfonts}
\usepackage{amsmath}
\usepackage{graphicx}
\usepackage{url}
\usepackage{eepic, epsfig, amsmath, amssymb, latexsym, setspace, subfig}
\usepackage{rotating}
\usepackage{amsfonts}
\usepackage{amsmath}
\usepackage{graphicx}
\usepackage{nicefrac}

\newtheorem{theorem}{Theorem}[section]

\newtheorem{corollary}[theorem]{Corollary}

\numberwithin{equation}{section}

\newcommand{\qed}{\rule{7pt}{7pt}}
\newenvironment{proof}{\noindent{\bf Proof}\hspace*{1em}}{\hfill\qed\vspace{0.125in}}

\newcommand{\x}{\mathbf{x}}
\newcommand{\y}{\mathbf{y}}
\newcommand{\z}{\mathbf{z}}
\newcommand{\w}{\mathbf{w}}
\newcommand{\e}{\mathbf{e}}

\def\bfe{{\mathbf e}}

\def\bfh{{\mathbf h}}

\def\bfv{{\mathbf v}}
\def\bfw{{\mathbf w}}
\def\bfx{{\mathbf x}}
\def\bfy{{\mathbf y}}
\def\bfz{{\mathbf z}}
\begin{document}
\title{On State Estimation with Bad Data Detection}

\author{Weiyu Xu, Meng Wang, and Ao Tang \\
    School of ECE, Cornell University, Ithaca, NY 14853}


%Telephone: (800) 555--1212, Fax: (888) 555--1212}
%\IEEEauthorblockA{\IEEEauthorrefmark{4}Tyrell Inc., 123 Replicant Street, Los Angeles, California 90210--4321}


\maketitle

\begin{abstract}
In this paper, we consider the problem of state estimation through observations possibly corrupted with both bad data and additive observation noises. A mixed $\ell_1$ and $\ell_2$ convex programming is used to separate both sparse bad data and additive noises from the observations. Through using the almost Euclidean property for a linear subspace, we derive a new performance bound for the state estimation error under sparse bad data and additive observation noises. Our main contribution is to provide sharp bounds on the almost Euclidean property of a linear subspace, using the ``escape-through-a-mesh'' theorem from geometric functional analysis. We also propose and numerically evaluate an iterative convex programming approach to performing bad data detections in nonlinear electrical power networks problems.
\end{abstract}

\IEEEpeerreviewmaketitle

\section{Introduction}
In this paper, we study the problem of state estimation under both bad data and observation noise. In state estimation problems, the observations may be corrupted with abnormally large measurement errors, called bad data, in addition to the usual additive observation noise. More specifically, suppose we want to estate the state $\x$ described by an $m$-dimensional real-numbered  vector, and we make $n$ measurements, then these measurements can be written as an $n$-dimensional vector $\y$, which is related to the state vector through the measurement equation
\begin{equation}
\y=h(\x)+\bfv+\e,
\label{eq: powermodel}
\end{equation}
where $h(\x)$ is a nonlinear function relating the measurement vector to the state vector, and $\bfv$ is the vector of measurement noise, and $\e$ is the vector of bad data imposed on the measurement. In this paper, we assume that $\bfv$ is an $m$-dimensional vector with i.i.d. zero mean Gaussian elements of variance $\sigma^2$. We also assume that $\e$ is a vector with only $k$ nonzero entries, and the nonzero entries can take arbitrary real-numbered values, reflecting the nature of bad data.

It is well known that Least Square (LS) method can be used to suppress the effect of observation noise on state estimations. In LS method, we try to find a vector $\x$ minimizing
\begin{equation}
\|\y-h(\x)\|_{2}.
\label{eq:ls}
\end{equation}
However, the LS method generally only works well when there are no bad data $\e$ corrupting the observation $\y$.

In this paper, a mixed least $\ell_1$ norm and least square convex programming is used to simultaneously detect bad data and subtract additive noises from the observations. In our theoretical analysis of the decoding performance, we assume $h(\x)$ is a linear transformation $H\x$ with $H$ as an $n \times m$ matrix with i.i.d. standard zero mean Gaussian entries. Through using the almost Euclidean property for the linear subspace generated by $H$, we derive a new performance bound for the state estimation error under sparse bad data and additive observation noises.  In our analysis,  using the ``escape-through-a-mesh'' theorem from geometric functional analysis \cite{Gordon}, we are able to significantly improve on the bounds for the almost Euclidean property of a linear subspace, which may be interesting in a more general mathematical setting. Compared with earlier analysis on the same optimization problem in \cite{CandesErrorCorrection}, the analysis is new using the almost Euclidean property rather than the restricted isometry conditions used in \cite{CandesErrorCorrection}, and we are able to give explicit bounds on the error performance, which is sharper than the analysis using the restricted isometry conditions in \cite{CandesErrorCorrection}.

Inspired by bad data detection methods for linear systems, we further propose an iterative convex programming approach to perform combined bad data detection and denoising  in nonlinear electrical power networks. The static state of an electric power network can be described by the vector of bus voltage magnitudes and angles in power networks. However, in smart grid power networks, the measurement of these quantities can be corrupted due to errors in the sensors, communication errors in transmitting the measurement results, and adversarial compromises of the meters. So the state estimation of power networks needs to detect, identify, and eliminate large measurement errors \cite{BC,MG,FB}. Since the probability of large measurement errors occurring is very small, it is reasonable to assume that bad data are only present in a small fraction of the available meter measurements results. So bad data detection in power networks can be viewed as a sparse error detection problem, which shares similar mathematical structures as sparse recoveries problem in compressive sensing \cite{CT1, CandesErrorCorrection}. However, this problem in power networks has several unique properties when compared with ordinary sparse error detection problem \cite{CT1}. In fact, $h(\x)$ in (\ref{eq: powermodel}) is a nonlinear mapping instead of a linear mapping in \cite{CandesErrorCorrection}. Our iterative convex programming based algorithms work is shown by numerical examples working well in this nonlinear setting. Compared with \cite{KV82}, which proposed to apply $\ell_1$ minimization in bad data detection in power networks, our approach offers a better decoding error performance when both bad data and additive observation noises are present. \cite{TL10a}\cite{TL10b} considered state estimations under malicious data attacks, and formulated the problem of state estimation under malicious attacks as a hypothesis testing problem by assuming a prior probability distribution on the state $\x$. In contrast, our approach does not rely on any prior information on the signal $\x$ itself, and the performance bounds hold for arbitrary state $\x$.

The rest of this paper is organized as follows. In Section \ref{sec:condition}, we introduce the convex programming to perform joint bad data detection and denoising,  and derive the  performance bound on the decoding error based on the almost Euclidean property of linear subspaces. In Section \ref{sec:boundingEuclidean}, a sharp bound on the almost Euclidean property is given through the ``escape-through-mesh'' theorem.  In Section \ref{sec:evaluating}, we will present explicit bounds on the decoding error. In Section \ref{sec:numerical}, we introduce our algorithm to perform bad data detection in nonlinear systems, and present simulation results of its performance in power networks.

\section{Bad Data Detection for Linear Systems}
\label{sec:condition}
In this section, we will introduce a convex programming formulation to do bad data detection in a linear systems, and give a characterization of its decoding error performance. In a linear system, the $n \times 1$ observation vector is $\y=H\x+\e+\bfv$, where $\x$ is the $m \times 1$ signal vector ($m<n$), $\e$ is a sparse error vector with $k$ nonzero elements, $\bfv$ is a noise vector with $\|\bfv\|_2 \leq \epsilon$. In what follows, we denote the part of any vector $\bfw$ over any index set $K'$ as $\bfw_{K'}$.

We solve the following optimization problem involving optimization variables $\x^*$ and $\z$,  and we then estimate the state $\x$ to be $\hat{\x}$, which is the optimizing value for $\x^*$.
\begin{eqnarray}
\min_{\x^*,\z}  &&\|\y-H\x^*-\z\|_{1},\nonumber \\
{\text{subject to}}&&\|\z\|_2 \leq \epsilon.
%&&~~~~~~~~~~~~~H^{T}z=0.
\label{eq:bus}
\end{eqnarray}
We are now ready to give the main theorem which bounds the decoding error performance of (\ref{eq:bus}).
\begin{theorem}
Let $\y$, $H$, $\x$, $\e$ and $\bfv$ are specified as above. Suppose that the minimum nonzero singular value of $H$ is $\sigma_{\text{min}}$. Let $C$ be a real number larger than $1$, and suppose that every vector $\bfw$ in the subspace generated by the matrix $H$ satisfies $C\|\bfw_K\|_{1} \leq \|\bfw_{\overline{K}}\|_{1}$ for any subset $K \subseteq \{1,2,...,n\}$ with cardinality $|K|\leq k$, where $k$ is an integer, and $\overline{K}=\{1,2,...,n\}\setminus K$. We also assume the subspace generated by $H$ satisfies the \emph{almost Euclidean} property for a constant $\alpha \leq 1$, namely
\begin{equation*}
\alpha \sqrt{n} \|\bfw\|_2 \leq \|\bfw\|_{1}
\end{equation*}
holds for every $\bfw$ in the subspace generated by $H$

Then the solution $\hat{\x}$ satisfies
\begin{equation}
\|\x-\hat{\x}\|_{2} \leq \frac{2(C+1)}{\sigma_{\text{min}} \alpha (C-1)}\epsilon.
\end{equation}
\label{thm:bound}
\end{theorem}

\begin{proof}
Suppose that one optimal solution set to (\ref{eq:bus}) is $(\hat{\bfx},\hat{\bfz})$. Since $\|z\|_{2} \leq \epsilon$, we have $\|\hat{\bfz}\|_{1} \leq \sqrt{n}\|\hat{\bfz}\|_{2} \leq  \sqrt{n} \epsilon $.

Since $\x^*=\x$ and $\bfz=\bfv$ is a feasible solution for (\ref{eq:bus}), then
\begin{eqnarray*}
&&\|\y- H\hat{\x}-\hat{\z}\|_{1}\\
&=&\|H(\x-\hat{\x})+\e+\bfv-\hat{\z}\|_{1}\\
&\leq& \|H(\x-\x)+\e+\bfv-\bfv\|_{1}\\
&=&\|\e\|_{1}.
\end{eqnarray*}

Applying the triangle inequality to $\|H(x-\hat{\x})+\e+\bfv-\hat{\z}\|_{1}$, we further obtain
\begin{equation*}
 \|H(\x-\hat{\x})+\e\|_{1}-\|\bfv\|_{1}-\|\hat{\bfz}\|_1 \leq \|\e\|_1.
\end{equation*}

Denoting $H(\x-\hat{\x})$ as $\w$, because $\e$ is supported on a set $K$ with cardinality $|K| \leq k$, by the triangle inequality for $\ell_1$ norm again,
\begin{equation*}
 \|\e\|_{1}-\|\w_{K}\|_1+\|\w_{\overline{K}}\|_1-\|\bfv\|_{1}-\|\hat{\z}\|_1\leq \|\e\|_1.
\end{equation*}

So we have
\begin{equation}
 -\|\w_{K}\|_1+\|\w_{\overline{K}}\|_1 \leq \|\hat{\z}\|_1+\|\bfv\|_{1} \leq 2\sqrt{n}\epsilon
\label{eq:differencebounded}
\end{equation}

With $C\|\w_K\|_{1} \leq \|\w_{\overline{K}}\|_{1}$, we know
\begin{equation*}
 \frac{C-1}{C+1} \|\w\|_{1} \leq -\|\w_{K}\|_1+\|\w_{\overline{K}}\|_1.
\end{equation*}

Combining this with (\ref{eq:differencebounded}), we obtain
\begin{equation*}
 \frac{C-1}{C+1} \|\w\|_{1} \leq 2\sqrt{n}\epsilon.
\end{equation*}

By the almost Euclidean property $\alpha \sqrt{n} \|\w\|_2 \leq \|\w\|_{1}$, it follows:
\begin{equation}
\|\w\|_{2} \leq \frac{2(C+1)}{\alpha (C-1)}\epsilon.
\label{eq:wl2norm}
\end{equation}

By the definition of singular values,
\begin{equation}
\sigma_{\text{min}} \|\x-\hat{\x}\|_2 \leq \|H(\x-\hat{\x})\|_2=\|\w\|_2,
\end{equation}
so combining (\ref{eq:wl2norm}), we get
\begin{equation*}
\|\x-\hat{\x}\|_{2} \leq \frac{2(C+1)}{\sigma_{\text{min}} \alpha (C-1)}\epsilon.
\end{equation*}
\end{proof}

Note that when there are no sparse errors present, the decoding error bound satisfies $\|\x-\hat{\x}\|_{2} \leq \frac{1}{\sigma_{\text{min}} }\epsilon$, Theorem \ref{thm:bound} shows that the decoding error of (\ref{eq:bus}) is  oblivious to the presence of bad data, no matter how large in amplitude these bad data can be. This phenomenon also observed in \cite{CandesErrorCorrection} by using the restricted isometry condition for compressive sensing.

We remark that, for given $\y$ and given $\epsilon$, by strong lagrange duality theory, the solution $\hat{\x}$ to (\ref{eq:bus}) will correspond to the solution to $\x$ in the following problem (\ref{eq:lambdaduality}) for some Lagrange duality variable $\lambda \geq 0$. As $\epsilon \geq 0$ increases, the corresponding $\lambda$ that produces the same solution to $\x$ will correspondingly decrease.
\begin{equation}\label{eq:lambdaduality}
\min_{\bfx, \bfz} \quad \|\bfy-H\bfx-\bfz\|_1 +\lambda \|\bfz\|_2.
\end{equation}
In fact, when $\lambda \rightarrow \infty$, (\ref{eq:lambdaduality})
approaches
\begin{equation}\nonumber
\min_{\bfx} \quad \|\bfy-H\bfx\|_1 ,
\end{equation}
and when $\lambda \rightarrow 0$,  (\ref{eq:lambdaduality})
approaches
\begin{equation}\nonumber
\min_{\bfx} \quad \|\bfy-H\bfx\|_2.
\end{equation}
Thus, (\ref{eq:lambdaduality}) can be viewed as a weighed version of
$\ell_1$ minimization and $\ell_2$ minimization (or equivalently the
LS method). We will later use numerical experiments to show that in
order to recover a sparse vector from measurements with both noise
and errors, this weighted version outperforms both $\ell_1$
minimization and the LS method.

 In the next two sections, we will aim at explicitly computing
$\frac{2(C+1)}{\sigma_{\text{min}} \alpha (C-1)} \times \sqrt{n}$,
which will denote $\varpi$ later in this paper. The appearance of
the $\sqrt{n}$ factor is to compensate for the energy scaling of
large random matrices and its meaning will be clear in later
context. Next, we will compute explicitly the almost Euclidean
property constant $\alpha$.

\section{Bounding the Almost Euclidean Property}
\label{sec:boundingEuclidean}
In this section, we would like to give a quantitative bound on the
almost Euclidean property constant $\alpha$ such that with high
probability (with respect to the measure for the subspace generated
by the random $H$), $\alpha \sqrt{n} \|\w\|_2 \leq \|\w\|_{1}$ holds
for every vector $\w$ from the subspace generated by $H$. Here we
assume that each element of $H$ is generated from the standard
Gaussian distribution $N(0,1)$. So the subspace generated by $H$ is
a uniformly distributed $(n-m)$-dimensional subspaces from the high dimensional geometry.

 To ensure that the subspace generated from $H$ satisfies the almost Euclidean property with $\alpha>0$, we must have the event that the subspace generated by $H$ does not intersect the set $\{\w\in S^{n-1}| \|\w\|_{1}< \alpha \sqrt{n} \|\w\|_2 \}$, where $S^{n-1}$ is the Euclidean sphere in $R^n$. To evaluate the probability that this event happens, we will need the following ``escape-through-mesh'' theorem.
\begin{theorem} \cite{Gordon}
\label{thm:escapethroughmesh}
Let $S$ be a subset of the unit Euclidean sphere $S^{n-1}$ in $R^n$. Let $Y$ be a random $m$-dimensional subspace of $R^{n}$, distributed uniformly in the Grassmanian with respect to the Haar measure. Let $w(S)=E(\sup_{\w\in S}(\bfh^T\bfw))$, where $\bfh$ is a random column vector in $R^{n}$ with i.i.d. $N(0,1)$ components. Assume that $w(S) < (\sqrt{n-m}-\frac{1}{2\sqrt{n-m}})$. Then
\begin{equation*}
P(Y \bigcap S=\emptyset)>1-3.5 e^{-\frac{(\sqrt{n-m}-\frac{1}{2\sqrt{n-m}})-w(S)}{18}}.
\end{equation*}
\end{theorem}

From Theorem \ref{thm:escapethroughmesh}, we can use the following programming to get an estimate of the upper bound of $w(\bfh, S)$. Because the set $\{\w\in S^{n-1}| \|\w\|_{1}< \alpha \sqrt{n} \|\w\|_2 \}$ is symmetric, without loss of generality, we assume that the elements of $\bfh$ follow i.i.d.  half-normal distributions, namely the distribution for the absolute value of a standard zero mean Gaussian random variables. With $h_i$ denoting the $i$-th element of $\bfh$, this is equivalent to
\begin{eqnarray}
\label{eq:maxproductoptimization}
\max && \sum_{i=1}^{n} h_{i} y_{i}\\
{\text{subject to}}&& y_{0} \geq 0, 1\leq i \leq n\\
&&\sum_{i=1}^{n}y_{i} \leq \alpha \sqrt{n}\\
&&\sum_{i=1}^{n} y_{i}^2=1.
\end{eqnarray}

Following the method from \cite{StojnicThresholds}, we use the Lagrange duality to find an upper bound for the objective function of (\ref{eq:maxproductoptimization}).
\begin{eqnarray}
&&\min_{u_{1} \geq 0, u_{2}\geq 0, \lambda \geq 0}\max_{w}\bfh^{T} \w-u_{1}(\sum_{i=1}^{n}w_i^2-1)\\
&&-u_2(\sum_{i=1}^{n}w_{i}-\alpha \sqrt{n})+\sum_{i=1}^{n} \lambda_{i} w_{i},
\label{eq:maxmin}
\end{eqnarray}
where $\lambda$ is a vector $(\lambda_{1}, \lambda_{2}, ..., \lambda_{n})$.

First, we maximize (\ref{eq:maxmin}) over $w_{i}$, $i=1,2, ..., n$ for fixed $u_{1}$, $u_{2}$ and $\lambda$. By making the derivatives to be zero, the minimizing $w_{i}$ is given by
\begin{equation*}
w_{i}=\frac{h_{i}+\lambda_{i}-u_{2}}{2u_{1}},  1\leq i \leq n
\end{equation*}

Plugging this back, we get
\begin{eqnarray}
 &&\bfh^{T} \w-u_{1}(\sum_{i=1}^{n}w_i^2-1)\\
 &&-u_2(\sum_{i=1}^{n}w_{i}-\alpha \sqrt{n})+\sum_{i=1}^{n} \lambda_{i} w_{i}\\
 &&=\frac{\sum_{i=1}^{n}{(-u_2+\lambda_{i}+h_i)^2}}{4 u_{1}}+u_1+\alpha \sqrt{n}u_2.
 \label{eq:insidemin}
\end{eqnarray}

Next, we minimize (\ref{eq:insidemin}) over $u_1 \geq 0$. It is not
hard to see the minimizing $u_{1}^*$ is
\begin{equation*}
 u_{1}^*=\frac{\sqrt{\sum_{i=1}^{n}{(-u_2+\lambda_{i}+h_i)^2}}}{2},
\end{equation*}
and the corresponding minimized value is
\begin{equation}\label{eqn:insidemin2}
{\sqrt{\sum_{i=1}^{n}{(-u_2+\lambda_{i}+h_i)^2}}}+\alpha \sqrt{n}
u_2.
\end{equation}

Then, we minimize (\ref{eqn:insidemin2}) over $\lambda \geq 0$. Given $\bfh$ and $u_2 \geq 0$, it is easy to see that the minimizing
$\lambda$ is

\[ \lambda_i = \left\{ \begin{array}{ll}
         u_2-h_i & \mbox{if $h_i \leq u_2$};\\
        0 & \mbox{otherwise},\end{array} \right. \]
%\begin{eqnarray*}
%\lambda_i  &=&u_2-h_i \quad \textrm{if } h_i\leq u_2 \\
%&=&0 \quad \textrm{otherwise}
%\end{eqnarray*}
and the corresponding minimized value is
\begin{equation}\label{eq:upperboundinstance}
\sqrt{\sum_{1 \leq i \leq n:\\ h_i <u_2}(u_2-h_i)^2}+\alpha
\sqrt{n}u_2.
\end{equation}

% Now if we take any $u_2 \geq 0$, % and $\lambda \geq 0$,
% we
%would have an upper bound for (\ref{eq:maxmin}):
%\begin{equation}
%{\sqrt{\sum_{i=1}^{n}{(-u_2+\lambda_{i}+h_i)^2}}}+\alpha \sqrt{n} u_2.
%\label{eq:upperboundinstance}
%\end{equation}

Now if we take any $u_2 \geq 0$, (\ref{eq:upperboundinstance}) serves
as an upper bound for (\ref{eq:maxmin}). Since $\sqrt{\cdot}$ is a
concave function, by Jensen's inequality, we have for any given $u_2
\geq 0$,
%\begin{equation*}
%E(\sup_{w\in S}(h^Tw)) \leq \sqrt{ E\{\sum_{i=1}^{n}{(-u_2+\lambda_{i}-h_i)^2}\} }+\alpha \sqrt{n} u_2
%\end{equation*}
%for any given $u_2 \geq 0$ and $\lambda \geq 0$.
\begin{equation}\label{eqn:upper}
E(\sup_{\w\in S}(\bfh^Tw)) \leq \sqrt{ E\{\sum_{1\leq i \leq n: h_i
<u_2}{(u_2-h_i)^2}\} }+\alpha \sqrt{n} u_2.
\end{equation}
%and $\lambda \geq 0$.
Since $\bfh$ has i.i.d. half-normal components, the righthand
side of (\ref{eqn:upper}) equals to
\begin{equation}\label{eqn:boundu2}
(\sqrt{(u_2^2+1)\textrm{erfc}(u_2/\sqrt{2})-\sqrt{2/\pi}u_2e^{-u_2^2/2}}+\alpha
u_2)\sqrt{n},
\end{equation}
where $\textrm{erfc}$ is the error function.

One can check that (\ref{eqn:boundu2}) is convex in $u_2$. Given
$\alpha$, we minimize (\ref{eqn:boundu2}) over $u_2 \geq 0$ and let
$g(\alpha)\sqrt{n}$ denote the minimum value. Then from
(\ref{eqn:upper}) and (\ref{eqn:boundu2}) we know
\begin{equation}
w(S)=E(\sup_{\w\in S}(\bfh^T\w)) \leq g(\alpha) \sqrt{n}.
\end{equation}
Given $\delta=\frac{m}{n}$, we pick the largest $\alpha^*$ such that
$g(\alpha^*) <\sqrt{1-\delta}$. Then as $n$ goes to infinity, it
holds that
\begin{equation}
w(S)\leq g(\alpha^*)\sqrt{n}<(\sqrt{n-m}-\frac{1}{2\sqrt{n-m}}).
\end{equation}
Then from Theorem \ref{thm:escapethroughmesh}, with high
probability  $\|\w\|_1 \geq \alpha^*\sqrt{n}\|\w\|_2$ holds for
every vector $\w$ in the subspace generated by $H$. We numerically
calculate how $\alpha^*$ changes over $\delta$ and plot the curve in
Fig. \ref{fig:alpha}. For example, when $\delta=0.5$,
$\alpha^*=0.332$, thus $\|\w\|_1 \geq 0.332\sqrt{n}\|\w\|_2$ for all
$\w$ in the subspace generated by $H$.

% matlabcode compressed sensing/power2
\begin{figure}[t] \centering
\includegraphics[scale=0.5]{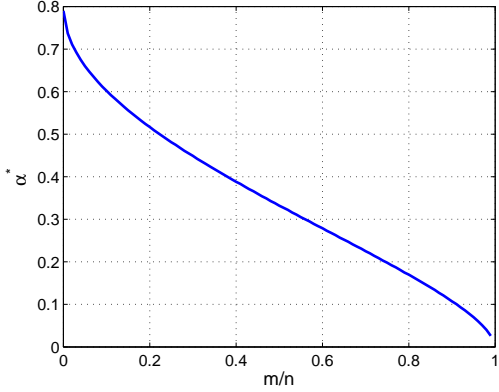}
%\end{psfrags}
\caption{$\alpha^*$ over $m/n$}\label{fig:alpha}
\end{figure}

Note that when $\frac{m}{n}=\frac{1}{2}$, we get $\alpha=0.332$. That is much larger than the known $\alpha$ used in \cite{Yin}, which is approximately $0.07$ (see Equation (12) in \cite{Yin}). When applied to the sparse recovery problem considered in \cite{Yin}, we will be able to recover any vector with no more than $0.0289n=0.0578m$ nonzero elements, which are $20$ times more than the $\frac{1}{384}m$ bound in \cite{Yin}.

 \section{Evaluating the Robust Error Correction Bound}
\label{sec:evaluating}
If the elements in the measurement matrix $H$ are i.i.d. as the unit
real Gaussian random variables $N(0,1)$, following upon the work of
Marchenko and Pastur \cite{Marcenko67}, Geman\cite{Geman80} and
Silverstein \cite{silver} proved that for $m/n=\delta$, as $n
\rightarrow \infty$, the smallest nonzero singular value
\begin{equation*}
\frac{1}{\sqrt{n}}\sigma_{\text{min}} \rightarrow 1-\sqrt{\delta}
\end{equation*}
almost surely as $n \rightarrow \infty$.

Now that we have already explicitly bounded $\alpha$ and $\sigma_{\text{min}}$, we now proceed to characterize $C$. It turns out that our earlier result on the almost Euclidean property can be used to computed $C$.
\begin{theorem}
 Suppose an $n$-dimensional vector $\w$ satisfies $\|\w\|_{1}\geq \alpha \sqrt{n} \|\w\|_2$. Then if for some set $K \subseteq \{1,2,...,n\}$ with cardinality $|K|=k \leq n$,
\begin{equation*}
\frac{\|\w_{K}\|_1}{\|\w\|_{1}}= \beta,
\end{equation*}
then $\beta$ must be a number satisfying
\begin{equation*}
\frac{\beta^2}{k}+\frac{(1-\beta)^2}{n-k} \leq \frac{1}{\alpha^2 n}
\end{equation*}
\label{thm:almostEuclideanB}
\end{theorem}

\begin{proof}
Without loss of generality, we let $\|\w\|_1=1$. Then by the Cauchy-Schwarz inequality,
\begin{eqnarray*}
\|\w\|_{2}^2&=&\|\w_{K}\|_{2}^2+\|\w_{\overline{K}}\|_{2}^2\\
&\geq& (\frac{\|\w_{K}\|_{1}}{\sqrt{k}})^2 + (\frac{\|\w_{\overline{K}}\|_{1}}{\sqrt{n-k}})^2\\
&=& (\frac{\beta^2}{k}+\frac{(1-\beta)^2}{n-k})\|\w\|_{1}^2.
\end{eqnarray*}

At the same time, by the almost Euclidean property,
\begin{equation*}
\alpha^2 n \|\w\|_{2}^2 \leq \|\w\|_{1}^2,
\end{equation*}
so we must have
\begin{equation*}
\frac{\beta^2}{k}+\frac{(1-\beta)^2}{n-k} \leq \frac{1}{\alpha^2 n}
\end{equation*}
\end{proof}

\begin{corollary}
If a nonzero $n$-dimensional vector $\w$ satisfies $\|\w\|_{1}\geq \alpha \sqrt{n} \|\w\|_2$, and for any set $K \subseteq \{1,2,...,n\}$ with cardinality $|K|=k \leq n$, if $C\|\w_{K}\|_{1}= \|\w_{\overline{K}}\|_1$ for some number $C\geq 1$, then
\begin{equation}
\frac{k}{n} \geq \frac{(B+1-C^2)-\sqrt{(B+1-C^2)^2-4B}}{2B},
\label{eq:largestsparsity}
\end{equation}
where $B=\frac{(C+1)^2}{\alpha^2}$.
\end{corollary}

\begin{proof}
If $C\|\w_{K}\|_{1}\geq \|\w_{\overline{K}}\|_1$, we have
\begin{equation*}
\frac{\|\w_{K}\|_1}{\|\w\|_{1}}= \frac{1}{C+1}.
\end{equation*}
So by Theorem \ref{thm:almostEuclideanB}, $\beta=\frac{1}{C+1}$ satisfies
\begin{equation*}
\frac{\beta^2}{k}+\frac{(1-\beta)^2}{n-k} \leq \frac{1}{\alpha^2 n}.
\end{equation*}

This is equivalent to
\begin{equation*}
\frac{1}{\frac{k}{n}}+\frac{C^2}{1-\frac{k}{n}} \leq \frac{(C+1)^2}{\alpha^2}
\end{equation*}

Solving this inequality for $\frac{k}{n}$, we get (\ref{eq:largestsparsity}).

\end{proof}

So for a sparsity $\frac{k}{n}$, this corollary can be used to find
$C$ such that $\frac{\|\w_{K}\|_1}{\|\w\|_{1}}= \frac{1}{C+1}$.
Combining these results on computing $\sigma_{\text{min}}$, $\alpha$
and $C$, we can then compute the bound
$\frac{2(C+1)}{\sigma_{\text{min}} \alpha (C-1)} \sqrt{n}=\varpi$ in
Theorem \ref{thm:bound}. For example, when
$\delta=\frac{m}{n}=\frac{1}{2}$, we plot the bound $\varpi$ as a
function of $\frac{k}{n}$ in Fig. \ref{fig:varpi}

\begin{figure}[t]
\centering
\includegraphics[scale=0.5]{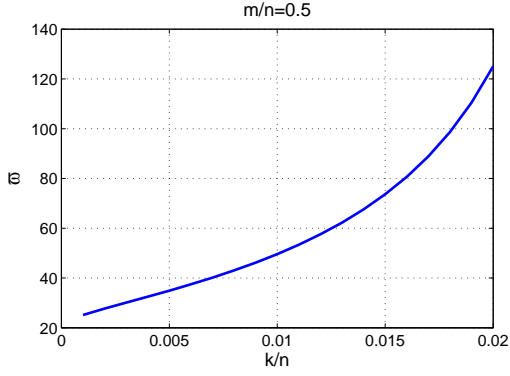}
%\end{psfrags}
\caption{$\varpi$ versus $\frac{k}{n}$}\label{fig:varpi}
\end{figure}

\section{Numerical Results}
\label{sec:numerical}
\textbf{Experiment 1:} We first consider recovering a signal vector from Gaussian
measurements. We generate the measurement matrix $H^{n \times m}$
with i.i.d. ${N}(0,1)$ entries and a vector $\bfx'\in
{R}^m$ with i.i.d Gaussian entries. Let
$\bfx=\bfx'/\|\bfx'\|_2$ be the signal vector. Let $m=60$ and
$n=150$. We first consider the recover performance when the number
of erroneous measurements is fixed. We randomly choose twelve
measurements and flip the signs of these measurements. For each
measurement, we also independently add a Gaussian noise from
${N}(0,\sigma^2)$. For a given $\sigma$, we apply
(\ref{eq:lambdaduality}) to estimate $\bfx$ using $\lambda$ from 0
to 13, and pick the best $\lambda^*$ with which the estimation error
is minimized. For each $\sigma$, the result is averaged over fifty
runs. Fig. \ref{fig:gaussianlambda2sigma} shows the curve of
$\lambda^*$ against $\sigma$. When the number of measurements with
bad data is fixed, $\lambda^*$ decreases as the noise level
increases.

\begin{figure}[t]
\centering
\includegraphics[scale=0.5]{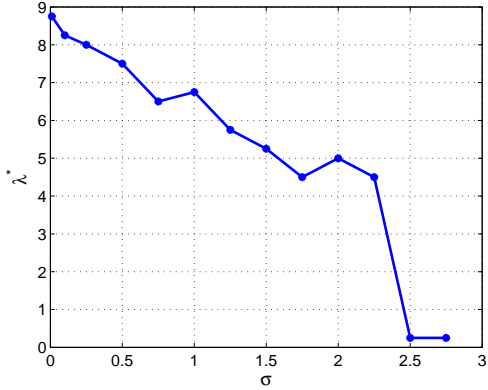}
%\end{psfrags}
\caption{$\lambda^*$ versus $\sigma$ for Gaussian
measurements}\label{fig:gaussianlambda2sigma}
\end{figure}

We next fix the noise level and consider the estimation performance
when the number of erroneous measurements changes. Each measurement
has a Gaussian noise independently drawn from ${N}(0,0.5^2)$. Let
$\rho$ denote the percentage of erroneous measurements. Given
$\rho$, we randomly choose $\rho n$ measurements, and each such
measurement is added with a Gaussian error independently drawn from
${N}(0,5^2)$. The estimation result is averaged over fifty runs.
Fig. \ref{fig:gaussianfixlambda} shows how the estimation error
changes as $\rho$ increases for different $\lambda$. $\lambda=8$ has
the best performance in this setup compared with a large value
$\lambda=15$ and a small value $\lambda=0.05$.

\begin{figure}[t]
\centering
\includegraphics[scale=0.5]{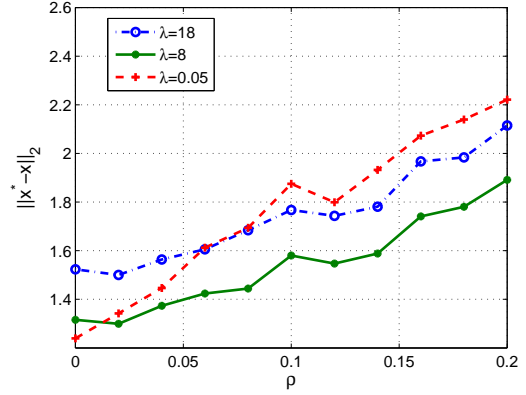}
%\end{psfrags}
\caption{$\lambda^*$ versus $\rho$ for Gaussian
measurements}\label{fig:gaussianfixlambda}
\end{figure}

\textbf{Experiment 2:} We also consider estimating the state of the power system from
available measurements and known system configuration. The state
variables are the voltage magnitudes and the voltage angles at each
bus. The measurements can be the real and reactive power injections
at each bus, and the real and reactive power flows on the lines. All
the measurements are corrupted with noise, and a small fraction of
the measurements contains errors. We would like to estimate the
state variables from the corrupted measurements. % using our proposed
%method.

The relationship between the measurements and the state variables
for a $k'$-bus system can be stated as follows \cite{KV82}:
\begin{eqnarray}\label{eqn:sys1}
P_i&=&\sum_{j=1}^{k'} E_i E_j Y_{ij} \cos(\theta_{ij} +\delta_i
-\delta_j),
\end{eqnarray}
\begin{eqnarray}
Q_i&=&\sum_{j=1}^{k'} E_i E_j Y_{ij} \sin(\theta_{ij} +\delta_i
-\delta_j),
\end{eqnarray}
\begin{eqnarray}
P_{ij}&=&E_i E_j Y_{ij} \cos(\theta_{ij} +\delta_i -\delta_j) \nonumber\\
&&-E_i^2 Y_{ij}\cos \theta_{ij} +E_i^2 Y_{si} \cos\theta_{si} \quad
i \neq j,
\end{eqnarray}
\begin{eqnarray}
Q_{ij}&=&E_i E_j Y_{ij} \sin(\theta_{ij} +\delta_i -\delta_j)
\nonumber \\&&-E_i^2 Y_{ij}\sin \theta_{ij} +E_i^2 Y_{si}
\sin\theta_{si} \quad i \neq j, \label{eqn:sys2}
\end{eqnarray}
where $P_i$ and $Q_i$ are the real and reactive power injection at
bus $i$ respectively, $P_{ij}$ and $Q_{ij}$ are the real and
reactive power flow from bus $i$ to bus $j$, $E_i$ and $\delta_i$
are the voltage magnitude and angle at bus $i$. $Y_{ij}$ and
$\theta_{ij}$ are the magnitude and phase angle of admittance from
bus $i$ to bus $j$, $Y_{si}$ and $\theta_{si}$ are the magnitude and
angle of the shunt admittance of line at bus $i$. Given a power
system, all $Y_{ij}$, $\theta_{ij}$, $Y_{si}$ and $\theta_{si}$ are
known.

For a $k'$-bus system, we treat one bus as the reference bus and set
the voltage angle at the reference
bus to be zero. %let $\bfx \in
%{R}^{2m=1}$ denote the vector of state variables
There are $m=2k'-1$ state variables with the first $k'$ variables for
the bus voltage magnitudes $E_i$ and the
rest $k'-1$ variables for the bus voltage angles $\theta_i$. % (assuming the
%voltage angle to be zero at one reference bus).
Let $\bfx \in {R}^{m}$ denote the state variables and let
$\bfy \in {R}^{n}$ denote the  $n$ measurements of the real
and reactive power injection and power flow.
 Let $\bfv \in {R}^n$ denote the noise
and $\bfe \in {R}^n$ denote the sparse error vector. Then we
can write the equations in a compact form,
\begin{equation}
\bfy=h(\bfx)+\bfv+\bfe,
\end{equation}
where $h(\cdot)$ denotes $n$ nonlinear functions defined in
(\ref{eqn:sys1}) to (\ref{eqn:sys2}).
%denotes the $n$ known nonlinear functions of power
%system, please refer to \cite{KV82} for explicit expressions.

An estimate of the state variables, $\hat{\bfx}$, can be obtained by
solving the following minimization problem,
\begin{equation}\label{eqn:nonlinear}
\min_{\bfx, \bfz} \quad \|\bfy-h(\bfx)-\bfz\|_1 +\lambda \|\bfz\|_2,
\end{equation}
where $\hat{\bfx}$ is the optimal solution $\bfx$. $\lambda>0$ is a fixed parameter. When $\lambda \rightarrow \infty$, (\ref{eqn:nonlinear}) approaches
\begin{equation}\label{eqn:l1}
\min_{\bfx} \quad \|\bfy-h(\bfx)\|_1 ,
\end{equation}
and when $\lambda \rightarrow 0$,  (\ref{eqn:nonlinear}) approaches
\begin{equation}\label{eqn:l2}
\min_{\bfx} \quad \|\bfy-h(\bfx)\|_2.
\end{equation}

Since $h$ is nonlinear, we linearize the equations and apply an
iterative procedure to obtain a solution. We start with the initial
state $\bfx^0$ where $x^0_i=1$ for all $i \in \{1,...,n\}$, and
$x^0_i=0$ for all $i \in \{n+1,..., 2n-1\}$. In the $k$th iteration,
let $\Delta \bfy^k=\bfy-h(\bfx^{k-1})$, then we solve the following
convex optimization problem,
\begin{equation}\label{eqn:linear}
\min_{\Delta\bfx, \bfz} \quad \|\Delta\bfy^k-H\Delta\bfx-\bfz\|_1
+\lambda \|\bfz\|_2,
\end{equation}
where $H^{n \times m}$ is the Jacobian matrix of $h$ evaluated at
$\bfx^{k-1}$. Let $\Delta\bfx^k$ denote the optimal solution
$\Delta\bfx$ to (\ref{eqn:linear}), then the state estimation is
updated by
\begin{equation}
\bfx^{k}=\bfx^{k-1}+\Delta\bfx^k.
\end{equation}
We repeat the process until $\Delta \bfx^k \rightarrow 0$.

\begin{figure}[t]
\centering
\includegraphics[scale=0.25]{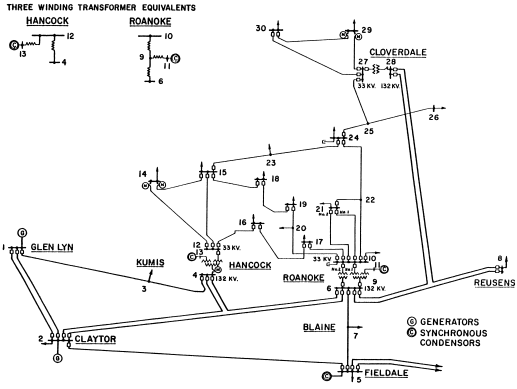}
%\end{psfrags}
\caption{IEEE 30-bus test system}\label{fig:bus}
\end{figure}

We evaluate the performance on the IEEE 30-bus test system. Fig.
\ref{fig:bus} shows the structure of the test system. Then the state
vector contains fifty-nine variables. We take one hundred
measurement including the real and reactive power injection at each
bus and some of the real and reactive power flows on the lines. We
first consider how the estimation performance changes as the noise
level increases when the erroneous measurements are fixed. The
errors of the measurements are simulated by inverting the sign of
the real power injection at bus 2, bus 3, bus 5, bus 26 and bus 30,
and inverting the sign of the reactive power injection at bus 30.
Each measurement also contains a Gaussian noise independently drawn
from ${N}(0,\sigma^2)$. For a fixed noise level $\sigma$, we
solve (\ref{eqn:nonlinear}) by the iterative procedure using
different $\lambda$ (from 0.5 to 12). The estimation performance is
measured by $\|\bfx^*-\hat{\bfx}\|_2$, where $\bfx^*$ is the true
state variable and $\hat{\bfx}$ is our estimation. For a fixed
$\sigma$, we choose the $\lambda^*$ to be the one with which
$\|\bfx^*-\hat{\bfx}\|_2$ is minimal among all the $\lambda$'s we
consider. The result is averaged over fifty runs. Fig.
\ref{fig:lambda2sigma} shows how $\lambda^*$ changes as $\sigma$
increases from 0 to 0.2. When the noise level is low, i.e. the
measurements basically  only contain errors, then the estimation
performance is better when we use a larger $\lambda$. When the noise
level is high, a smaller $\lambda$ leads to a better performance.
\begin{figure}[t]
\centering
\includegraphics[scale=0.5]{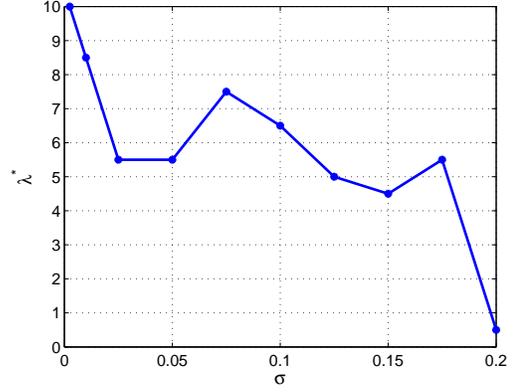}
%\end{psfrags}
\caption{$\lambda^*$ versus $\sigma$}\label{fig:lambda2sigma}
\end{figure}

We also study how the estimation performance changes as the number
of erroneous measurements increases. Each of the one hundred
measurements contains random Gaussian noise independently drawn from
${N}(0,0.05^2)$. Let $\rho$ denote the percentage of erroneous
measurements with bad data. For fixed $\rho$, we randomly choose the
set $T$ of erroneous measurements with cardinality $|T|=\rho m$.
Each erroneous measurement contains an additional Gaussian error
independently drawn from ${N}(0,0.7^2)$. We than calculate the
solution $\hat{\bfx}$ of (\ref{eqn:nonlinear}) and the estimation
error $\|\bfx^*-\hat{\bfx}\|_2$. Fig. \ref{fig:fixlambda} shows how
the estimation error changes as $\rho$ increases. The results are
averaged over fifty runs. When $\lambda$ is small ($\lambda=0.5$),
(\ref{eqn:nonlinear}) approaches (\ref{eqn:l1}), and the estimation
error is relatively large if $\rho$ is small, i.e. the measurements
basically contain only noise. When $\lambda$ is large
($\lambda=12$), (\ref{eqn:nonlinear}) approaches (\ref{eqn:l2}), and
the estimation error is relatively large if $\rho$ is large, i.e.
the measurements contains errors besides noise. In contrast, if we
choose $\lambda$ to be 7 in this case, the estimation error is
relatively small for all $\rho$ among the three choices of
$\lambda$.
\begin{figure}[t]
\centering
\includegraphics[scale=0.5]{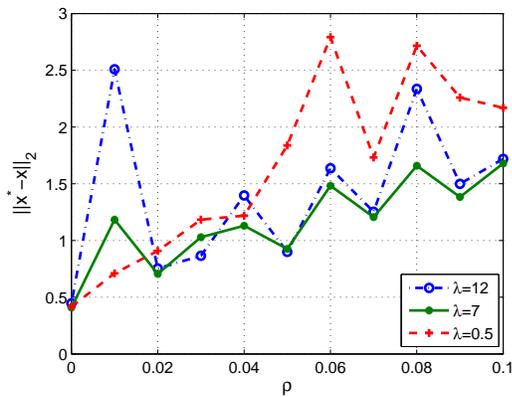}
%\end{psfrags}
\caption{The estimation error versus $\rho$}\label{fig:fixlambda}
\end{figure}

\section{Conclusion}\label{sec:conclusion}
In this paper, we study state estimation through observations corrupted with both bad data and additive observation noises. A mixed $\ell_1$ and $\ell_2$ convex programming is used to separate both sparse bad data and additive noises from the observations.
We used the almost Euclidean property of a linear subspace to provide sharp bounds on this convex programming based state estimation method. We also give sharp bounds for the almost Euclidean property of a linear subspace using the ``escape-through-a-mesh'' theorem from geometric functional analysis \cite{Gordon}. We then propose an iterative convex programming based methods to perform state estimation with bad data detection in the nonlinear electrical power network problems. Simulation results confirm the effectiveness of the algorithms in denoising and detecting bad data at the same time.
%iterative convex programming approach to performing bad data detections in nonlinear electrical power networks problems.


\section*{Acknowledgment}
The research is supported by NSF under CCF-0835706 and ONR under N00014-11-1-0131.

\bibliographystyle{IEEEbib}

\end{document}